\documentclass[preprint,12pt]{elsarticle}

\usepackage{amssymb}

\usepackage{ulem,bbm}

\usepackage{xcolor}

\usepackage[margin=2.5cm]{geometry}

\usepackage{tabularx,lipsum,graphicx,tabularray}  

\usepackage[hyphens]{url}
\usepackage[colorlinks=true, urlcolor=blue, pdfborder={0 0 0}]{hyperref}
\hypersetup{
  colorlinks   = true, 
  urlcolor     = blue, 
  linkcolor    = blue, 
  citecolor   = red 
}
\hypersetup{breaklinks=true}
\usepackage{subcaption}
\usepackage{booktabs}
\usepackage{multirow}
\usepackage{tabularx}	
\usepackage{color,colortbl}

\usepackage[draft]{changes} 

\usepackage{algorithmic}
\usepackage{algorithm}

\usepackage{graphicx}
\usepackage{mathtools}
\usepackage{amsthm}
\usepackage{dsfont}

\biboptions{sort&compress}

\usepackage{pgfplots}
\usepackage{tikz}
\usepgfplotslibrary{groupplots}
\usetikzlibrary{matrix} 
\pgfplotsset{compat=1.17}
\usepgfplotslibrary{dateplot}

\newtheorem{remark}{Remark}

\newtheorem{example}{Example}

\newtheorem{definition}{Definition}[section]
\newtheorem{prop}{Proposition}[section]
\newtheorem{lem}{Lemma}[section]
\newtheorem{thm}{Theorem}[section]
\newtheorem{rem}{Remark}[section]

\usepackage{tikz}

\def\PP{\mathcal{P}}
\def\HH{\mathcal{H}}
\def\pk{\texttt{pk}}
\def\seed{\texttt{seed}}
\def\index{\texttt{index}}
\def\MM{\mathcal{M}}
\def\cumsum{\texttt{cumsum}}
\def\ww{\vec{w}}
\def\ss{\vec{s}}
\def\SS{\vec{S}}

\def\AA{\mathcal{A}}

\def\Alg{\texttt{Alg}}
\def\committee{\mathcal{M}}
\def\MM{\committee}
\def\bbN{\mathbb{N}}
\def\bbR{\mathbb{R}}
\newcommand{\aps}[1]{\vert #1 \vert}

\newcommand{\norm}[1]{\| #1 \|}

\def\E{\mathbb{E}}
\def\P{\mathbb{P}}
\def\dynamicProgramming{\texttt{D}}

\newif\ifcomment
\commenttrue
\newcommand{\np}[1]{\textcolor{blue}{#1 - Nikita}}
\newcommand{\seb}[1]{\textcolor{orange}{#1 - Seb}}

\renewcommand{\added}{}
\makeatletter
\newcommand{\ion}[2]{%
  #1$\;$%
  \if b\expandafter\@car\f@series\relax\@nil
    \begingroup 
      \sbox0{\rmfamily\mdseries\textsc{v}}%
      \resizebox{!}{\ht0}{\rmfamily\@Roman{#2}}%
    \endgroup
  \else
    \textsc{\rmfamily\@roman{#2}}%
  \fi
}
\makeatother

\makeatletter
\def\@author#1{\g@addto@macro\elsauthors{\normalsize%
    \def\baselinestretch{1}%
    \upshape\authorsep#1\unskip\textsuperscript{%
      \ifx\@fnmark\@empty\else\unskip\sep\@fnmark\let\sep=,\fi
      \ifx\@corref\@empty\else\unskip\sep\@corref\let\sep=,\fi
      }%
    \def\authorsep{\unskip,\space}%
    \global\let\@fnmark\@empty
    \global\let\@corref\@empty  
    \global\let\sep\@empty}%
    \@eadauthor={#1}
}
\makeatother

\journal{Blockchain: Research and Applications}

\begin{document}

\begin{frontmatter}



\title{
Deterministic Bounds in Committee Selection: Enhancing Decentralization and Scalability in Distributed Ledgers}




\author[B4B]{Grigorii Melnikov}
\ead{g.melnikov@b4b.world}
\affiliation[B4B]{organization={B4B.World},
            country={Montenegro}}

\author[AMU]{Sebastian Mueller}
\ead{sebastian.muller@univ-amu.fr}
\affiliation[AMU]{organization={Aix Marseille Universit{é}, CNRS, Marseille},
            country={France}}

\author[IF]{Nikita Polyanskii}
\ead{{Corresponding author*}{nikita.polyansky@gmail.com}}
\affiliation[IF]{organization={IOTA Foundation},
            city={Berlin},
            country={Germany}}

\author[Sk,HSE]{Yury Yanovich}
\ead{y.yanovich@skoltech.ru}
\affiliation[Sk]{organization={Skolkovo Institute of Science and Technology},
            city={Moscow},
            country={Russia}}
\affiliation[HSE]{organization={HSE University}, city={Moscow},
            country={Russia}}

\begin{abstract}
Consensus plays a crucial role in distributed ledger systems, impacting both scalability and decentralization. Many blockchain systems use a weighted lottery based on a scarce resource such as a stake, storage, memory, or computing power to select a committee whose members drive the consensus and are responsible for adding new information to the ledger. Therefore, ensuring a robust and fair committee selection process is essential for maintaining security, efficiency, and decentralization. 

There are two main approaches to randomized committee selection. In one approach, each validator candidate locally checks whether they are elected to the committee and reveals their proof during the consensus phase. In contrast, in the second approach, a sortition algorithm decides a fixed-sized committee that is globally verified. This paper focuses on the later approach, with cryptographic sortition as a method for fair committee selection that guarantees a constant committee size. Our goal is to develop deterministic guarantees that strengthen decentralization. We introduce novel methods that provide deterministic bounds on the influence of adversaries within the committee, as evidenced by numerical experiments. This approach overcomes the limitations of existing protocols that only offer probabilistic guarantees, often providing large committees that are impractical for many quorum-based applications like atomic broadcast and randomness beacon protocols.

\end{abstract}


\end{frontmatter}

\section{Introduction}\label{sec:intro}



The quest for an effective consensus mechanism, a cornerstone of distributed ledger technologies (DLTs), remains a critical challenge in this rapidly evolving landscape \cite{Capossele2021,He2021}. Consensus protocols, which ensure agreement on the ledger's state across all participants, are pivotal to the functioning of DLTs. However, as networks scale, balancing scalability, security, and decentralization involves significant trade-offs \cite{Guo2022}.
One practical avenue to address scalability is the delegation of validation tasks to a smaller, representative committee drawn from the network. \added{This approach is employed by systems such as Ethereum and Cardano~\cite{Pavloff2023,Kiayias2017,Kusmierz2021}.} This approach is designed to boost efficiency, while the methods for selecting such committees often involve weighted lotteries or cryptographic sortition, tailored to maximize decentralization. 

Through the execution of simple computations and verifications involving the ledger's records, all network participants can agree on the composition of the elected committee~\cite{gilad2017algorand}. The fundamental concept behind this solution is that a relatively small committee enables the network to reach a consensus more efficiently, thereby reducing communication costs and latency. \added{Our current research focuses on committee selection, leaving the voting process within the committee beyond the scope of this paper.}

There are various approaches to committee selection that can classified as follows. Let us assume that we have $N$ participants, each with a certain weight $w_i$, representing some scarce resources. A locally verifiable approach is similar to a lottery, where every participant creates their own lot whose input is a globally verifiable seed, for instance by dint of verifiable random functions (VRFs)~\cite{micali1999verifiable}. If this lot is a ``winner'' (e.g., the VRF output is above a certain threshold), the participant becomes a committee member. The advantage of this approach is that only local information is needed and the events of different participants entering the committee are independent. However, a negative aspect is that the committee size is not constant, and we cannot guarantee that a certain weight is present in the committee. Moreover, in the case of rotating committees, this does not allow as to speak of fixed seats, a property that might be desirable for the smooth transitions between different committees. 

In contrast, the (weighted) round-robin selection mechanism offers simplicity and predictability. Participants take turns in a fixed order, ensuring every node becomes part of the committee proportional to their weight $w_i$. While this method ensures a predictable and fair distribution of participation over time, it is vulnerable to targeted attacks due to its predictability. \added{Additionally, member selection is interdependent and thus less likely to be locally verifiable.}

In this paper, we aim to combine the advantages of these two classes and build a random sortition that allows a fair distribution over time, a fixed committee size, and develop deterministic guarantees that enhance decentralization and resilience against adversarial influence.
Consequently, our protocols are not locally verifiable, but instead, the whole committee is clear to every participant who knows global randomness.

\added{This study makes several contributions to fixed-size committee selection--sortition, including:}
\begin{enumerate}
    \item 
        \added{A formulation of fairness and decentralization properties of the procedure. Fairness is a binary requirement for the algorithms, while decentralization is a quantitative property, where a higher value of decentralization is preferable.}
    \item 
        \added{A proposition that sortition with certain decentralization provides deterministic guarantees for an honest majority of the chosen committee. Recall that existing sortition methods don't usually provide any deterministic decentralization properties and, thus, only offer probabilistic guarantees for an honest majority.}
    \item
        \added{A proposal of four fair sortition algorithms with a comparative analysis of their decentralization.}
\end{enumerate}
\added{We challenge researchers to design an algorithm that achieves optimal decentralization.}

The remainder of the paper is organized as follows. Section~\ref{sec: related work} discusses existing sortition protocols, focusing on the notion of fairness in the context of committee selection. Section~\ref{sec: model} introduces notation and two key definitions of fairness and decentralization. In Section~\ref{sec: proposed algorithms}, we present four novel sortition algorithms that achieve fairness and different levels of decentralization.
Finally, Section~\ref{sec: experiments} provides an experimental evaluation of the proposed algorithms.

\section{Related Work}\label{sec: related work}
The fundamental principle of sortition, a political system that dates back to the Athenian democracy of the 5th century BC, involves the random selection of individuals. These individuals form a sortition panel tasked with making decisions on behalf of the broader population from which they were chosen. This ancient practice was based on the idea that randomness could prevent corruption and ensure a fair representation of the population.
In the modern era, the principle of sortition has transcended its political origins and found practical applications in various fields, including distributed ledger technologies (DLTs). As DLTs, such as blockchain systems, grapple with challenges like scalability, security, and decentralization, the concept of sortition is used to create efficient and more scalable consensus mechanisms.

\subsection{Sortition in blockchain protocols}
Algorand's blockchain protocol~\cite{gilad2017algorand} pioneers the use of cryptographic sortition for committee selection, a significant innovation in the field. This method utilizes a Verifiable Random Function (VRF), detailed in~\cite{micali1999verifiable}, for randomly selecting users for block proposals and voting, which is contingent on their stake. In each round, the algorithm assigns to a user holding stake \( m \in \mathbb{N} \) a sample from a binomial random variable \( B(m, \tau) \). This binomial then determines the number of ``seats'' of the participant or its weight in the committee.
The parameter \( \tau \)  regulates the expected committee size or the committee weight, which is equal to $\tau$ times the total sum of the weights. Note that every participant gets a binomial number of seats in the committee with expectation proportional to its stake. Thus, splitting or merging among different identities does not influence the number of seats in the committee, or equivalently, the sortition protocol is fair. 

Filecoin~\cite{wang2023security, filecoin_spec} adopts a similar approach for block proposer selection. The key distinction lies in the utilization of pledged storage instead of stake for determining weights in the lottery and the employment of Poisson random variables, instead of binomial ones, to model the weight of block proposers. Specifically, a participant with weight $m$ gets weight $Pois(m \tau)$ when proposing a block (with $\tau$ being a global constant). This design choice is driven by the aim of simplifying implementation while retaining the robustness against merging and splitting. 

The block proposer selection in Cardano's Ouroboros algorithm~\cite{david2018ouroboros,badertscher2018ouroboros} bears resemblance to these methods. Here, a participant with a share $\alpha$ of stake is elected to propose a block in a slot with probability $1-(1-\tau)^{\alpha}$ (with $0<\tau<1$ being a global parameter). While this approach does not offer the fairness characteristic of Algorand or Filecoin, it provides a weaker version of the robustness property: the likelihood of at least one block proposer emerging is independent of how the stake is distributed among the participants and equal to~$\tau$. 

In Ethereum 2.0~\cite{Gasper}, the committee selection process is an integral part of its Proof-of-Stake consensus mechanism. Time is divided into slots and epochs. There are several roles and committees that are distributed in a random and verifiable way. The approach is different from the “lottery”-based proposal. At the start of each epoch,  validators are shuffled using a permutation function based on a randomness seed. This shuffling determines which validators will be grouped into committees for the upcoming epoch. The probability of being on a committee does not depend on the amount of staked tokens; every elected member has the same influence. Fairness comes from the restriction of allowing only validators with the same weight of 32 ETH  effective stake.

\subsection{Other Fairness Notions in Sortition}
The paper~\cite{amoussou2019fairness} discusses the fairness of sortition in the context of reward distributions, i.e., a blockchain protocol is called \textit{fair} if any honest node that has a fraction $\alpha$ of the total weight in the system will get
at least $\alpha$ fraction of the total reward that is given in the system (over the infinite execution of the system). Fairness is shown to be achievable only if the network model is synchronous or eventually synchronous and Byzantine (faulty or malicious) processes are identifiable.

The research~\cite{flanigan2020neutralizing} addresses the challenge of self-selection bias in sortition panels, where only a fraction of invited individuals are willing to participate, and those who participate may not be representative of the population. The proposed algorithm aims for end-to-end fairness, ensuring all members of the population appear on the panel with close-to-equal probability while also satisfying quotas for demographic representation. This work emphasizes the importance of fairness at both individual and group levels, accounting for intersectional groups that might be overlooked in traditional quota systems.

\subsection{Other Committee Selection Procedures}
There is a big class of voting-based committee selection procedures employed in many blockchains such as EOS, Cosmos Hub, Polkadot, Tron, Binance Smart Chain, Tezos, ICON, Lisk, BitShares, SUI. In such a procedure, each participant gives their vote to a subset of candidates. Then by properly evaluating votes, a committee is formed by taking the set of the most approved candidates. See more details on such voting-based models in~\cite{benhaim2023scaling}.
\section{Model and Properties}\label{sec: model}
\subsection{Notations}

We consider a set of $N$ participants, where the $i$th participant has a positive weight $w_i$. Denote $\vec{w} = (w_1, \dots, w_N)^T$ and assume that weights are normalized, i.e. $\sum_{n=1}^N w_i = 1$.

A sortition is a committee selection algorithm that takes as input $\vec{w}$ and outputs a random subset of participants with assigned voting power. In this paper, we fix the desired size $M$ of a committee and our proposed sortition algorithms output $M$-sized committees. The randomness in a sortition algorithm comes from a publicly known random $seed$.\footnote{The way how randomness is chosen and used in an algorithm is out of the scope of the paper.}  In the following, we write $\AA$ to denote a sortition algorithm. The selected committee is written as $\MM= \{ m_1, m_2, \ldots, m_{M} \}$ and is interpreted as a random variable on the set of all subsets of $[N]$ of size $M$.
We write $\P$ and $\E$ for the underlying probability space and its expectation. 
The voting power of the $i$th committee member is denoted as $g(m_i; \MM)$.  While we provide sortition algorithms that first finds $\MM$ and then independently assign the weight $g_{n}$ to every participant $n$, we assume that the actual voting power $g(m_i; \MM)$ is normalized $\sum_{m \in \MM} g(m_i; \MM) = 1$, i.e.,  $g(m_i; \MM) = \frac{g_{m_i}}{\sum_{j=1}^M g_{m_j}}$ for $i \in [M]$, and $g(n; \MM) = 0$ for $n \not\in \MM$. We denote the vector of $g_n$'s as $\vec{g} = (g_1, \dots, g_N)^T$. 

\subsection{Desired Properties}
The most fundamental property of any committee selection procedure is a fair distribution of voting power.

\begin{definition}[Fairness]
A sortition algorithm is called \textbf{fair}, if the expected voting power of every participant is equal to its initial weight, i.e., for all $N\in \mathbb{N}$ and any given weight vector $\vec{w}$, it holds for all $n\in [N]$
\begin{eqnarray*}
    \E \left[ g(n; \MM) \right] = w_n.
\end{eqnarray*}
\end{definition}

The above fairness ensures that a participant's representation in the committee is equal to their initial weight on average. This fairness is not only a fundamental property but also crucial for maintaining stability within the system. It ensures that there is no incentive to manipulate voting power by splitting or merging weights.

Since the fairness definition depends only on expectations, it is not guaranteed for any realization of the committee it holds $g(n; \MM) = w_n$. Therefore, we introduce an additional property, that holds true when no participant possesses a voting power exceeding $1/\lambda$ times their weight, where $\lambda$ is a parameter. This constraint holds deterministically, ensuring that the voting power remains proportionally bounded in every committee.

\begin{definition}[Decentralization]
\label{def: decentralization}
A sortition algorithm is called $\lambda$-\textbf{decentralized} for $\lambda > 0$ if, for every committee realization $\MM$, the voting power of every committee member can not be larger than its initial weight multiplied by $1/\lambda$:
\begin{eqnarray*}
   \max_{i\in[M]} \frac{g(i; \MM)}{w_{i}} \le \frac{1}{\lambda}.
\end{eqnarray*}
\end{definition}
\begin{example}
     In the special case of equal initial weights $w_n=1/N$, and a sortition algorithm that assigns voting power $g_{m_i}=1/M$ for a random $M$-sized subset of participants, the decentralization equals $\lambda=\frac{1/N}{1/M}=\frac{M}{N}$. We observe that a larger committee size increases the decentralization, reaching its maximum of $\lambda=1$ for $M=N$. In the other extreme case, $M=1$, the decentralization is minimal with $\lambda=1/N.$
\end{example}

While many existing sortition algorithms provide \textit{probabilistic} guarantees for the honest majority of the committee, the above {decentralization} property gives \textit{deterministic} guarantees in case of a small initial weight of the adversary. 

\begin{prop}\label{prop:hornestMajority}
For any $\lambda$-decentralized sortition algorithm, 
the adversary can not get a majority of the total voting power if it initially controls at most  $\frac{\lambda}{2}$ of the initial weight.
\end{prop}
\begin{proof}
By Def.~\ref{def: decentralization}, the adversary with at most $\frac{\lambda}{2}$ of the initial weight could get at most  $\frac{\lambda}{2} \cdot \lambda^{-1} = 1/2$ of the voting power.
\end{proof}

The larger the value of $\lambda$, the more advantageous it becomes for the adversary to have any chance of achieving a majority in the resulting committee.


\section{Proposed Algorithms}
\label{sec: proposed algorithms}
In this section, we present four algorithms, all of which provide fairness, but guarantee different levels of decentralization.

\subsection{Algorithm \texttt{Stitch}}
We start with the most basic fair committee selection procedure that provides a committee of fixed size $M$. The idea of the proposed Alg.~\ref{algorithm:Stitch}, called \texttt{Stitch}, is to assign each participant a distinct interval in the interval $[0, 1]$ according to its initial weight and create a \textit{stitch} with a random starting point and fixed $1/M$-length intervals. The committee $\MM$ is formed by taking participants whose intervals intersect with the stitch. If the initial weight of each participant is bounded by $1/M$, then the committee will not contain any duplicating participants. The optional permutation in line 1 of Alg.~\ref{algorithm:Stitch} is not necessary for the property discussed in this paper. However, it does ensure a certain ``independence'' in the committee's composition. 

\begin{algorithm}[t]
\caption{\texttt{Stitch}}
\begin{algorithmic}[1]
\label{algorithm:Stitch}
\REQUIRE $\vec{w}$, $M$
\ENSURE $\vec{g}$, $\MM$
\STATE (Opt.) Apply a random permutation of the set of participants to $\vec{w}=(w_1,\ldots, w_N)$.
\STATE Divide the $[0,1]$-interval into subintervals of length $w_i$. 
\STATE Sample a point $x$ with uniform distribution in the $[0,1]$-interval.
\STATE Form the committee $\MM$ by taking all members whose intervals contain one of the points of the form $\{x+i/M\}$ with $i\in [M]$, where $\{\cdot \}$ denotes the fractional part.
\STATE Set $g_i\gets 1$ for $i \in \MM$ and $g_i \gets 0$ for $i \not\in \MM$.
\RETURN $\vec{g}$, $\MM$
\end{algorithmic}
\end{algorithm}

\begin{prop}\label{prop:fairnessStitch}
Let the weight vector $\vec{w}=(w_1,\ldots,w_{N})$ satisfy $\max\limits_{n \in [N]} w_n < 1/M$, where $M$ is the desired committee size. Then, the sortition algorithm \texttt{Stitch} is fair, and the committee is of size $M$.
\end{prop}

\begin{proof}
The initial point \(x\) is chosen uniformly from the interval \([0, 1]\). The sequence \(\{x + i/M\}\) for \(i \in [M]\) cycles through the interval \([0, 1]\).

We consider a specific weight interval of length \(w_n\) for participant \(n\) and suppose it has the form \([a, a+w_n]\) within \([0, 1]\). 
While \(x\) is uniformly distributed, the subsequent points are uniformly spread at a fixed distance of \(1/M\) apart from each other, starting from \(x\). Thus, the interval \([a, a+w_n]\) is intersected by the points \(\{x + i/M\}\) based on where \(x\) initially falls. Specifically, a point \(x + i/M\) falls into \([a, a+w_n]\) if \(x\) is in \(\{a - i/M, a + w_n - i/M\}\). Thus, the expected number of points falling into \(w_n\) from  the shifts \(\{x + i/M\}\) is $M \cdot w_n$. Since $w_n <1/M$ for all $n$, no participant is chosen twice and hence
\[
\E[g(n; \mathcal{M})] = M \cdot w_n \cdot \frac{1}{M} = w_n.
\]
\end{proof}
\begin{prop}\label{prop:decentralizationStitch}
Let the weight vector $\vec{w}=(w_1,\ldots,w_{N})$ satisfy $\max\limits_{n \in [N]} w_n < 1/M$. Then  the sortition algorithm \texttt{Stitch} (Alg.~\ref{algorithm:Stitch}) is $\lambda$-decentralized with $\lambda = M \min\limits_{i \in [N]} w_i$. 
\end{prop}
\begin{proof}
Once elected, a participant gets $1/M$ of the voting power. Therefore, we obtain
\begin{eqnarray*}
    \frac{1}{\lambda} = \max_{i\in[M]} \frac{g(m_i; \MM)}{w_{m_i}} = \frac{1/M}{\min\limits_{i\in [M]} w_{m_i}}
\end{eqnarray*}
or $\lambda = M \min\limits_{i\in [M]} w_{m_i}$.
\end{proof}

\subsection{Algorithm \texttt{Cumulative Rejection Sampling}}

The idea behind Alg. \ref{algorithm:CumulativeRejectionSampling}, called \texttt{Cumulative Rejection Sampling}, is to sample a random $M$-sized committee in a way such that we increase the likelihood of selecting committees with larger initial weights.
However, this modification introduces a bias which must be counteracted through a recalibration of the probabilities to be chosen in the committee or in the voting power $g_i$ to maintain the overall fairness of the algorithm. The Cumulative Rejection Sampling follows the first possibility. 

\begin{algorithm}[t]
\caption{\texttt{Cumulative Rejection Sampling}}
\begin{algorithmic}[1]
\label{algorithm:CumulativeRejectionSampling}
\REQUIRE $\vec{w}$, $M$
\ENSURE $\vec{g}$, $\mathcal{M}$
\STATE Based on $\vec{w}$ and $M$, compute $p_i$'s using \eqref{eq: pis}.
\STATE Generate a uniform random subset of $\MM = \{m_1, \dots, m_M\}$ participants.
\STATE Generate a uniform random variable $U \sim U[0, 1]$. If $U \geq \sum_{j=1}^M p_{m_j}$, reject the subset $\MM$ and return to Step 2.
\STATE Set $g_m=1$ for $m \in \MM$ and $g_m \gets 0$ for $m \not\in \MM$.
\RETURN $\vec{g}$, $\mathcal{M}$
\end{algorithmic}
\end{algorithm}

\begin{remark}
   One can replace $1$ in the uniform random variable $U[0, 1]$ in Step 2 with a constant defined as the maximum possible weight of an $M$-sized subset of $\vec{w}$.
\end{remark}

\begin{remark}
There are efficient algorithms with complexity $O(n)$ to sample from a uniform random subset, e.g., the Fisher-Yates shuffling~\cite{fisher-yates-wikipedia}.
\end{remark}


\begin{prop}\label{prop:CumulativeRejectionSampling}
Suppose $N>M>2$ and  each entry $w_i$ of the weight vector $\vec{w}$ satisfies 
\begin{equation}\label{eq:condition  wi}
    w_i \in  \left[\frac{M-1}{M(N-1)}, \frac{N-2M+M^2}{M^2(N-1)}\right]. 
\end{equation}
If one sets
\begin{equation}\label{eq: pis}
   p_i=\frac{M(N-1)}{N-M}w_i - \frac{M-1}{N-M},
\end{equation}
then the sortition algorithm \texttt{Cumulative Rejection Sampling} (Alg.~\ref{algorithm:CumulativeRejectionSampling}) is fair.
\end{prop}

\begin{proof}
The algorithm uses a standard rejection sampling technique. We want to sample from a distribution on all subsets of size $M$ with a probability proportional to the sum of its weights $p_i$; specifically, a certain committee $\MM=\{m_1,\ldots,m_M\}$ is chosen with probability $D\sum_{j=1}^M p_{m_j}$ with $D$ being a global constant (independent of $\MM$). The weights $p_i$'s will be specified later to ensure the fairness of the procedure. 

As we set $g_i=1$ and $g(i;\MM)=\frac{1}{M}$ when $i\in \MM$, the expected voting power of the $i$th participant equals the probability of the $i$th participant being selected to the committee multiplied by $1/M$. To achieve fairness, the first equality in the following chain of equalities must hold:
\begin{eqnarray*}
    w_i&=&\E[g(i;\MM)]\\
    &=& \P(i \in \MM)\frac{1}{M} \\&=& \sum_{\substack{m_1, \dots, m_M \\ i \in \{m_1, \dots, m_M\}}} \frac{D}{M}\sum_{i=1}^M p_{m_i} \\ &=&  \frac{D}{M}\binom{N-1}{M-1} p_i + \frac{D}{M}\sum_{j \neq i} \binom{N-2}{M-2} p_j \\ &=& \frac{D}{M}\binom{N-1}{M-1} p_i +  \frac{D}{M}\binom{N-2}{M-2} (1 - p_i),
\end{eqnarray*}
where we used in the last equality that $\sum_j p_j=1$.
Recall that the sum of $w_i$'s is equal to one. Thus
$$
1=\sum_{i} w_i = \frac{D}{M}\binom{N-1}{M-1} \sum_i p_i +  \frac{D}{M}\binom{N-2}{M-2} \sum_i(1 - p_i)
$$
Again recalling that $\sum_i p_i =1$ we can identify $D$:
\begin{eqnarray*}
   1= \frac{D}{M} \binom{N-1}{M-1}+ \frac{D}{M} \binom{N-2}{M-2}(N-1) = \frac{D}{M}\binom{N-1}{M-1}+ \frac{D}{M} \binom{N-1}{M-1} (M-1)
\end{eqnarray*}
or
$$
D = \frac{1}{\binom{N-1}{M-1}}.
$$
Using the obtained $D$, we can get the formula for $p_i$ as a function of $w_i$ as follows:
$$
w_i= \frac{p_i}{M}+\frac{(1-p_i)(M-1)}{M(N-1)}
$$
or
$$
p_i = \frac{M(N-1)}{N-M}w_i - \frac{M-1}{N-M}.
$$
Finally, we obtain the restrictions for $w_i$. Since $p_i$ must be positive, it follows that 
$$
w_i \ge \frac{M-1}{M(N-1)}. 
$$
Similarly, one can ensure $\sum_{j=1}^{M}p_{m_j}\le 1$ by imposing $p_i\le \frac{1}{M}$ for all $i$, thus,
$$
w_i\le \frac{N-M}{M^2(N-1)} + \frac{M-1}{M(N-1)}= \frac{N-2M+M^2}{M^2(N-1)}.
$$
\end{proof}
\begin{rem}
Note that the previous proposition also covers the case with homogeneous weights $w_i=1/N.$ For this case, the result follows by a standard symmetry argument. However, our main result here is that we can extend the fairness to a slightly larger class of weights. 
\end{rem}

Similar to Prop.~\ref{prop:decentralizationStitch}, we obtain the following decentralization result.

\begin{prop}\label{prop:decentralizationCRS}
Let the weight vector $\vec{w}=(w_1,\ldots,w_{N})$ satisfy \eqref{eq:condition  wi}. Then  the  \texttt{Weighted Rejection Sampling} (Alg.~\ref{algorithm:WeightedRejectionSampling}) is $\lambda$-decentralized with $\lambda = M \min\limits_{i \in [N]} w_i$. 
\end{prop}
\begin{remark}
To have good decentralization properties for both   \texttt{Stitch} and \texttt{Weighted Rejection Sampling}, one needs to control the eligibility of participants allowed to be included in a committee; specifically, their minimum weight has to be bounded below by a certain constant.
\end{remark}
\subsection{Algorithm \texttt{Weighted Rejection Sampling}}

The idea of Alg. \ref{algorithm:WeightedRejectionSampling}, called \texttt{Weighted Rejection Sampling}, is to sample $M$-sized committees and provide decentralization by rejecting committees with extremely low total initial weights.

\begin{algorithm}[t]
\caption{\texttt{Weighted Rejection Sampling}}
\begin{algorithmic}[1]
\label{algorithm:WeightedRejectionSampling}
\REQUIRE $\vec{w}$, $M$, $\alpha$
\ENSURE $\vec{g}$, $\mathcal{M}$
\STATE Based on $\vec{w}$, $M$ and $\alpha$, compute $p_i$'s using~\eqref{eqn:g-WRS}.
\STATE Generate a uniform random subset of $M$ participants $\committee = \{m_1, \dots, m_M\}$.
\STATE Generate a uniform random variable $U \sim U[0, 1]$. If $U \geq \sum_{j=1}^M p_{m_j} \cdot \mathbbm{1}\left(\sum_{j=1}^M w_{m_j} \geq \alpha\right)$, reject the subset $\committee$ and return to Step 2.
\STATE Set $g_{m}=p_{m}$ for $m \in \MM$ and $g_m = 0$ for $m \not\in \MM$.
\RETURN $\mathcal{M}$, $\vec{g}$
\end{algorithmic}
\end{algorithm}


Given a real parameter $\alpha$, for each participant $n\in[N]$, define the number of $M$-sized committees with member $n$ with  total initial weight at least $\alpha$ as follows
\begin{eqnarray}
\label{eqn:Cn-WRS}
    C_n(\alpha) = \left|\left\{\committee:\quad n \in \committee, \,\sum_{m \in \committee} w_m \geq \alpha \right\}\right|.
\end{eqnarray} 

\begin{prop}\label{prop:WeightedRejectionSampling}
Suppose $N>M>2$, $0<\alpha<1$ and for every $i \in [N]$, it holds that $C_i(\alpha)>0$. If one sets 
\begin{eqnarray}
\label{eqn:g-WRS}
    p_i = \frac{w_i/C_i(\alpha)}{\sum_{j=1}^N w_j/C_j(\alpha)},
\end{eqnarray}
then the sortition \texttt{Weighted Rejection Sampling} (Alg. \ref{algorithm:WeightedRejectionSampling}) is fair. 
\end{prop}
\begin{proof}
In this algorithm, we sample from a distribution on all subsets of size $M$ with total weight more than $\alpha$ such that a probability distribution is proportional to the sum of its $p_i$'s; specifically, a certain committee $\MM=\{m_1,\ldots,m_M\}$ is chosen with probability $C\cdot\mathbbm{1}\left(\sum_{j=1}^m w_{m_j}\ge \alpha\right)\sum_{j=1}^M p_{m_j}$ with $C$ being a global constant (independent of $\MM$).  Additionally, we can assume that $\sum_i p_i = 1$. To make the sortition algorithm fair, one needs to achieve the first equality in the following chain of equalities for all $i \in [N]$:
\begin{eqnarray*}
    w_i&=& \mathbb{E}\left[ g(i; \MM)\right] \\
    &=& \sum_{\committee: \; i \in \committee} g(i; \MM) \cdot P(\committee) \\
    &=& \sum_{\committee: \; i \in \committee} \frac{p_i}{\sum_{j=1}^M p_{m_j}} \cdot \sum_{j=1}^M p_{m_j} \cdot \mathbbm{1}\left(\sum_{j=1}^{M} w_{m_j} \geq \alpha\right) \cdot C \\
    &=& \sum_{\committee: \; i \in \committee} p_i \cdot \mathbbm{1}\left(\sum_{j=1}^M w_{m_j} \geq \alpha\right) \cdot C \\
    &=& C \cdot p_i \sum_{\committee: \; i \in \committee} \mathbbm{1}\left(\sum_{j=1} w_{m_j} \geq \alpha\right).
\end{eqnarray*}
 As the sum of indicators in the right-hand side of the above expression is equal to $C_i(\alpha)$, $p_i$'s are proportional to $w_i/C_i(\alpha)$. To have the property $\sum_i p_i = 1$, we set $p_i$'s by~\eqref{eqn:g-WRS}. 
\end{proof}

While the computation of $C_n(\alpha)$ is impractical in a general case, we show that it is feasible in some special cases.

\begin{prop}\label{prop:decentralizationWeightedRejectionSampling}
Suppose that before normalization, the entries of the weight vector $\vec{w}$ are positive integers. Then it is possible to compute all $p_i$'s using~\eqref{eqn:g-WRS} in $O(NMV)$ time, where $V = \alpha \sum_i w_i$.
\end{prop}
\begin{proof}




We shall prove that there is a dynamic programming algorithm for this problem that runs in $O(MVN)$ time. For $i\ge k\ge 0$, let $\dynamicProgramming[v, k, i]$ denote the number of $k$-sized subsets of coordinates in $\vec{w}$ that sum up to $v$ such that these $k$ elements are chosen among the first $i$ elements of vector $\vec{w}$. We note that the recursive relation holds
$$
\dynamicProgramming[v, k, i]=
\begin{cases}
    \text { 1 } & \text { if } v=k=0; \\ 
    \text{ 0 } & \text { if } v\neq 0,  k=0;\\
    \mathbbm{1}(\sum_{j=1}^kw_j =v )  & \text { if }  k=i\neq 0; \\
    \dynamicProgramming[v, k, i-1] 
    + 
    \dynamicProgramming[v-w_i, k-1, i-1] & \text { otherwise.}
\end{cases}
$$
We observe that the coefficient at $x^v y^m$, $v < V, m \le M$ in the polynomial
$$
d(x,y) 
= \prod_{i=1}^{N} (1 + x^{w_i} y) \mod y^{M + 1} \mod x^{V}
= \sum_{m=0}^{M} d_m(x)y^m, \;
\mathrm{deg}(d_{m}) < V
$$
equals the number of $m$-sized subsets of coordinates in $\vec{w}$ that sum up to $v$. This implies that evaluating $d_m(1)$ corresponds to the number of $m$-sized subsets whose sum is less than $V$.

To calculate the number of $m$-sized subsets with a sum less than $V$ that exclude $i$th element, we need to evaluate $q_m(1)$, where polynomial $q_m(x)$ is defined in the right-hand side below:
$$
q(x,y) = \prod_{j\in[N]\setminus\{i\}}(1 + x^{w_j} y) \mod y^{M + 1} \mod x^{V}
= \sum_{m=0}^{M} q_m(x)y^m, \;
\mathrm{deg}(q_{m}) < V
$$
By simple algebraic manipulation, we derive the following identifies
$$
(1 + x^{w_i}y) \sum_{m=0}^{M} q_m(x)y^m \mod y^{M+1} \mod x^V = \sum_{l=0}^{M} d_l(x)y^l \\
$$
$$
q_m(x) = d_m(x) - x^{w_i}q_{m-1}(x) \mod x^{V}
$$
For $m=M$, we obtain
$$
q_M(x) = \sum_{m} (-1)^m x^{m w_i} d_{M - m}(x) \mod x^{V},
$$
thus,
$$
q_M(1) = \sum_{m=0}^{M} (-1)^m \sum_{v=0}^{V - m w_i - 1} \dynamicProgramming[v, M - m, N].
$$

Define $\hat{C}(\alpha)$ as the number of $M$-sized sets of coordinates in $\vec{w}$ whose sum is less than $V$ and $\hat{C}_i(\alpha)$ as the number of such sets that include the $i$th element of $\vec{w}$. It holds that
$$
\hat{C}(\alpha) = \sum_{v=0}^{V-1} \dynamicProgramming[v, M, N]
$$
and
$$
\hat{C}_i(\alpha) = \hat{C}(\alpha) - q_M(1).
$$

Finally, we need to calculate $C_i(\alpha)$ for every $i \in [N]$ in order to compute $p_i$ as defined by~(\ref{eqn:g-WRS})
$$C_i(\alpha) = \binom{N - 1}{M - 1} - \hat{C}_i(\alpha).$$
\end{proof}

\begin{prop}\label{prop:decentralizationWRS}
Let $p_i$ be defined by (\ref{eqn:g-WRS}). Then the algorithm \texttt{Weighted Rejection Sampling}  (Alg.~\ref{algorithm:WeightedRejectionSampling}) is $\lambda$-decentralized with $\lambda \ge \min\limits_{i \in [N]} w_i(1 + \frac{F_i}{p_i})$, where $F_i$ is the sum of $M - 1$ smallest $p_j$'s, excluding $p_i$.
\end{prop}
\begin{proof}
    The proof follows directly from the definition of decentralization. The voting power of participant $i$ could become at most $p_i$ divided by the smallest sum of an $M$-sized committee that includes $p_i$.
\end{proof}

\subsection{Algorithm \texttt{Representative Electoral College}}

The idea of Alg.~\ref{algorithm:RepresentativeElectoralCollege}, called \texttt{Representative Electoral College}, is to partition participants into $M$ groups based on their initial weights, determine voting power for each group, and select one participant from each group. In order to ensure fairness, it is important to control the likelihood of a participant being chosen to the committee.

\begin{algorithm}[t]
\caption{\texttt{Representative Electoral College}}
\begin{algorithmic}[1]
\label{algorithm:RepresentativeElectoralCollege}
\REQUIRE $\vec{w}$, $M$
\ENSURE $\vec{g}$, $\mathcal{M}$
\STATE Sort participants in non-decreasing order of their initial weights, i.e., $w_1 \leq \dots \leq w_N$.
\STATE Split $N$ participants into $M$ almost equal groups: $G_1$ consists of first $N_1 = \lfloor N/M\rfloor + \mathbbm{1}(N \mod M \geq 1)$ participants, \dots, $G_i$ consists of next $N_i = \lfloor N/M\rfloor  + \mathbbm{1}(N\mod M \geq i)$ participants, \dots, $G_M$ consists of the last $N_M = \lfloor N/M\rfloor$ participants.
\STATE Based on $\vec{w}$ and the partition $G_i$'s compute $p_i$'s using~\eqref{eqn:g-REC}.
\STATE Select $\MM=\{m_1,\ldots,m_M\}$ randomly such that $m_i$ is drawn from $G_i$ and for each $j\in G_i$, $\P(m_i=j)=\frac{w_j}{\sum_{n \in G_i} w_n}$.
\STATE Assign $g_{m_i}=p_i$ for all $i\in [M]$ and $g_m=0$ for $m\not\in\MM$.
\RETURN $\vec{g}$, $\mathcal{M}$
\end{algorithmic}
\end{algorithm}


\begin{prop}\label{prop:RepresentativeElectoralCollege}
Suppose that the entries of $\vec{w}$ are sorted in non-decreasing order and the groups $G_i$'s are defined as in Alg.~\ref{algorithm:RepresentativeElectoralCollege}. If one sets
\begin{eqnarray}
\label{eqn:g-REC}
    p_m = \sum_{i \in G_m} w_i,
\end{eqnarray}
then the sortition \texttt{Representative Electoral College} (Alg.~\ref{algorithm:RepresentativeElectoralCollege}) is fair. 
\end{prop}
\begin{proof}
Consider participant $j$ from 
group $G_m$. Then its expected voting power can be computed as follows:
\begin{eqnarray*}
    \E[ g(j;\MM)] &=& p_m \cdot \P(j \in \MM) \\
    &=&
    \sum_{n \in G_m} w_n \cdot \frac{w_n}{\sum_{n \in G_m} w_n} \\ 
    &=& w_j.
\end{eqnarray*}
\end{proof}

\begin{prop}\label{prop:RepresentativeElectoralCollege}
The sortition \texttt{Representative Electoral College} (Alg.~\ref{algorithm:RepresentativeElectoralCollege}) is $\lambda$-decentralized with $\lambda = \min\limits_{m \in [M]}\min\limits_{i \in G_m} \frac{ w_i}{p_m}$, where $p_m$ is defined at~\eqref{eqn:g-REC}.
\end{prop}
\begin{proof}
Once elected to the committee, participant $j$ from group $G_m$ gets voting power $\frac{p_m}{\sum_i p_i}=p_m$. This implies the statement.
\end{proof}

\begin{figure}[h]
    \centering
    \includegraphics[width=0.75\linewidth]{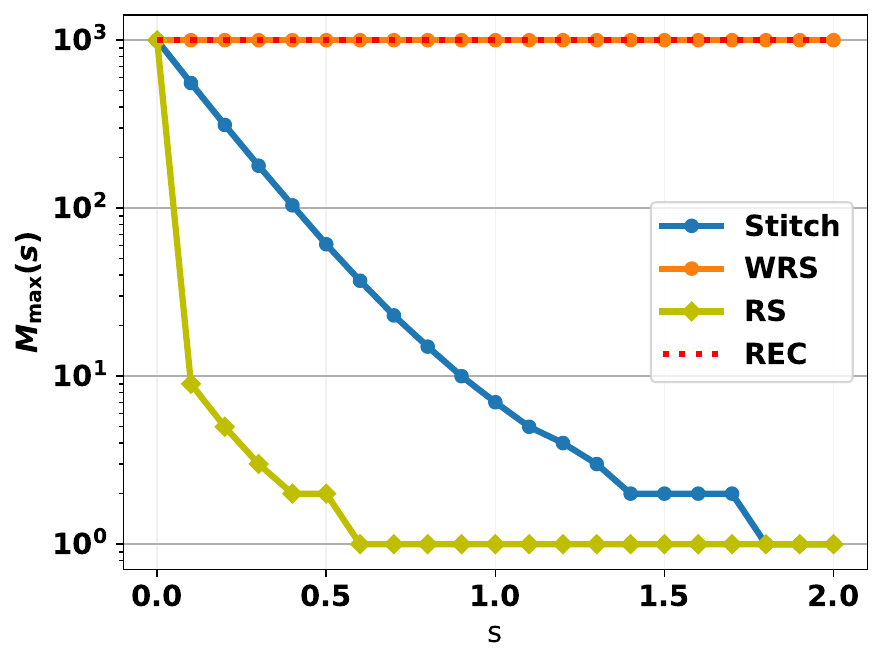}
    \caption{Maximum allowed $M$ as a function of $s$}
    \label{fig:max_M}
\end{figure}

\begin{figure}[h]
    \centering
    \includegraphics[width=0.75\linewidth]{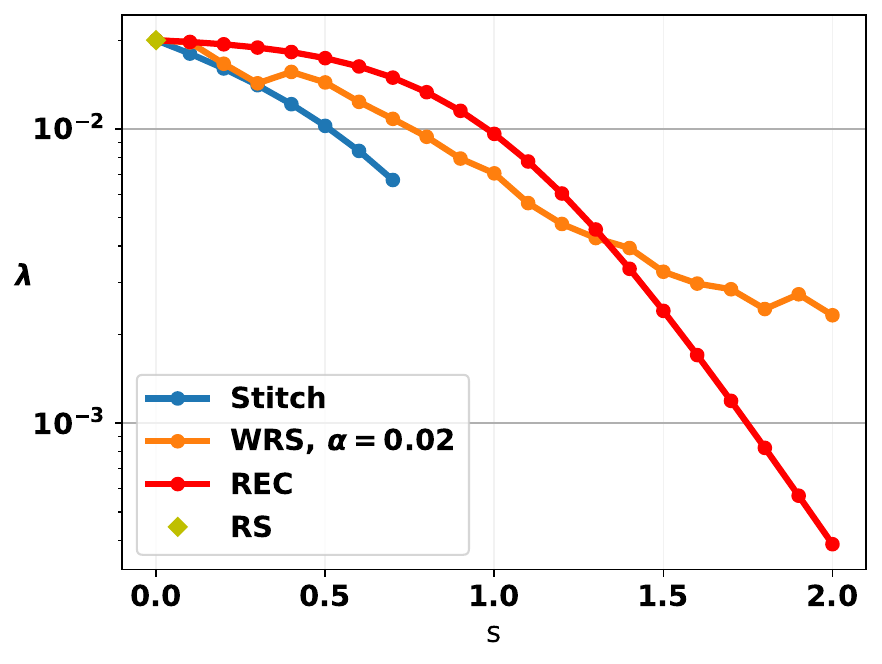}
    \caption{Decentralization as a function of $s$}
    \label{fig:M20}
\end{figure}

\begin{figure}[h]
    \centering
    \includegraphics[width=0.75\linewidth]{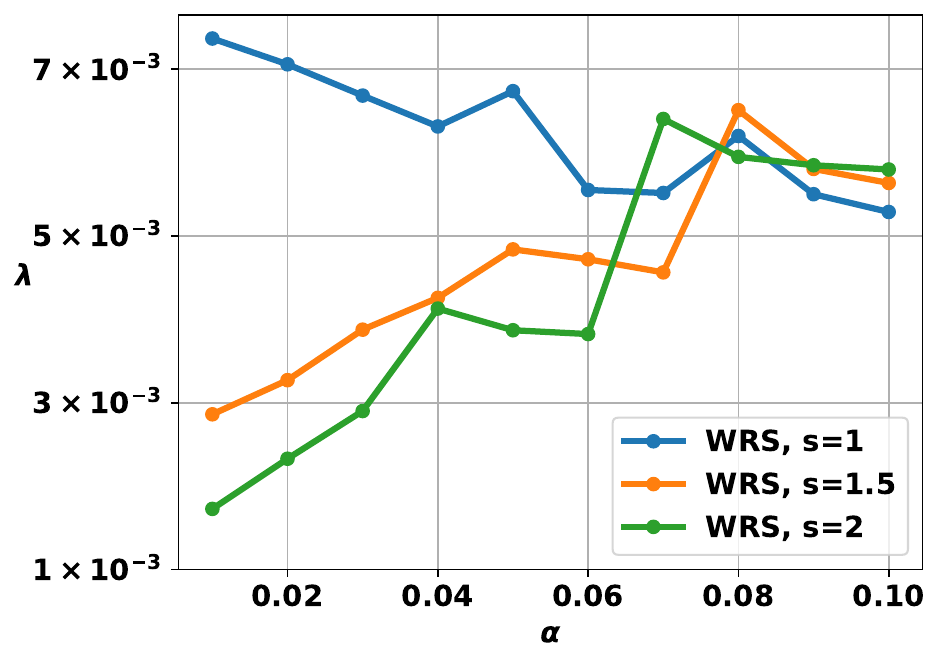}
    \caption{Decentralization as a function of $\alpha$}
    \label{fig:WRS-alpha}
\end{figure}

\section{Numerical Experiments}\label{sec: experiments}

In Section \ref{sec: proposed algorithms}, we have introduced four algorithms, all designed to achieve fairness while differing in their levels of decentralization. To compare these algorithms, we conducted numerical experiments. The experiments are carried out under the following settings. We have a total of $N=1000$ participants. The size $M$ of a committee $\MM$ varies, with a particular focus on $M=20$. This choice of $N$ and $M$ aligns with typical cases in existing blockchains, where $N$ typically ranges from hundreds to millions, and $M$ varies from tens to hundreds~\cite{DashCoreGroup2018,Buterin2020,Ovezik2022}.

We adopt Zipf's law \cite{Powers1998} to assign weights to participants; specifically, the non-normalized weight of participant $i$ is
$w_i = 1/i^s$,  where  $s$ is taken from the range $\{0, 0.1, \dots, 1.9, 2\}$. \texttt{Weighted Rejection Sampling} requires positive integer initial weights. In this case, we assign the non-normalized weight $w_i =  \lfloor N^s / i^s \rceil$. For $s=0$, all participants have equal weights, while larger values of $s$ introduce more centralization into the system.

\texttt{Stitch} and \texttt{Cumulative Rejection Sampling} (\texttt{CRS}) algorithms impose a constraint on $M$ (Propositions \ref{prop:fairnessStitch} and \ref{prop:CumulativeRejectionSampling})), establishing an upper limit for $M$ based on the initial weights. In contrast, \texttt{Weighted Rejection Sampling} (\texttt{WRS}) and \texttt{Representative Electoral College} (\texttt{REC}) algorithms do not restrict the possible values of $M$. Initially, we calculated the maximum permissible $M$ as a function of $s$, which is denoted as $M_{\text{max}}(s)$ in~Figure \ref{fig:max_M}. For \texttt{Stitch}, $M \geq 20$ is allowed until $s=0.7$, while for \texttt{CRS}, $M \geq 20$ is only valid at $s=0$, and from $s=0.6$ onwards, the only feasible choice for $M$ is $1$.

The decentralization $\lambda$ as a function of $s$ is given in Figure \ref{fig:M20}. All data points correspond to $N=1000$ and $M=20$. \texttt{Stitch} lacks data points for larger $M$ values due to restrictions. \texttt{CRS} features a single data point. \texttt{Stitch} demonstrates reasonable decentralization with a straightforward structure. \texttt{REC} yields better outcomes for small $s$, while \texttt{WRS} excels for larger~$s$.

\texttt{WRS} has an extra parameter $\alpha$ with Figure \ref{fig:WRS-alpha} displaying decentralization $\alpha$ as a function of $\alpha$. By fine-tuning $\alpha$, enhanced decentralization can be achieved.


Among the algorithms, \texttt{WRS} stands out as computationally intensive due to its pseudo-polynomial time complexity $O(MNV)$ and dependency on the $\alpha$-fraction of initial weights sum, $V=\alpha \sum_{i \in [N]} w_i$. According to \cite{kusmierz2022centralized}, the most relevant Zipf's law coefficients for tokens and coins are typically within the range of 0.5 to 1.5.  Our experiments showed that even on a laptop system equipped with 16GB RAM and a 3.4 GHz ARM CPU, a non-optimized, single-core implementation of \texttt{WRS} on weights with Zipf's law coefficients from this interval takes less than 2 seconds to calculate all $p_i$ on the $\alpha$ parameter with the highest $\lambda$-decentralisation.

\section{Conclusions}
The selection of committees in distributed ledger systems plays a pivotal role in shaping scalability, security, and decentralization of these systems. This paper has delved deep into the crucial need for a robust and equitable committee selection process to uphold the system's integrity. Our focus has centered on cryptographic sortition, which ensures a fixed committee size. In this method, the initial weights of participants—stakes--and the desired committee size serve as input, generating a randomized subset of specified size with weighted elements where these weights translate into voting power.

By introducing innovative sortition algorithms that prioritize fairness and varying levels of decentralization, this paper enriches the existing knowledge base in this field. We have formally defined two fundamental properties crucial for such selections: fairness, ensuring a balanced representation proportional to initial weights, and decentralization, curbing the influence of minor stakeholders to prevent disproportionate voting power. While fairness is a binary measure, decentralization is quantitatively assessed. Our research demonstrates that adversaries with limited total weight cannot dominate any committee, with the permissible stake of adversaries directly linked to the level of decentralization.

We have presented four cryptographic sortition algorithms that address fairness concerns and devised computational approaches to quantify their decentralization. Through numerical experiments conducted on simulated yet realistic data inspired by DLTs, we have showcased the efficacy of these algorithms. Among them, the Stitch algorithm offers reasonable decentralization but struggles with highly skewed initial weight distributions. The Cumulative Rejection Sampling algorithm excels only in cases of equal initial weights, while the Representative Electoral College and Weighted Rejection Sampling algorithms stand out for their superior decentralization in scenarios of balanced and imbalanced initial weights, respectively.

Looking ahead, our future endeavors will focus on developing an algorithm that maximizes decentralization under fairness constraints. The practical application of our theoretical findings and algorithmic solutions holds immense potential in Proof-of-Stake-based DLTs, paving the way for enhanced system efficiency and integrity.

\section*{\added{Acknowledgments}}

\added{During the preparation of this work the authors used ChatGPT in order to enhance the English language and readability of the paper. After using this tool/service, the authors reviewed and edited the content as needed and take full responsibility for the content of the publication.}

\bibliographystyle{elsarticle-num}
\bibliography{bib}
\end{document}